\documentclass[12pt,a4paper,oneside]{article}
\usepackage[left=2.5cm, right=2.5cm, top=2cm, bottom=2.5cm]{geometry}
\usepackage{graphicx} 
\usepackage{natbib}
\usepackage{url}
\usepackage{booktabs}
\usepackage{datetime}
\usepackage{float}
\usepackage{amsmath}
\usepackage{amssymb}
\usepackage{amsthm}
\newtheorem{assumption}{Assumption}
\newtheorem{theorem}{Theorem}
\newtheorem{corollary}{Corollary}
\newtheorem{definition}{Definition}
\newtheorem{remark}{Remark}
\newcommand{\citeg}[1]{\citep[see, e.g.,][]{#1}}
\usepackage[colorlinks=true,
            citecolor=blue,
            linkcolor=blue,
            urlcolor=blue]{hyperref}
\author{
Marina Telezhkina
\thanks{
\textit{Corresponding author.}
Email: m.s.telezhkina@gmail.com
}
 \\
Pforzheim University 
\and
Dominik Wied \\
University of Cologne
}          

\date{\monthname[\month] \the\year}

\begin{document}
\title{Cosine metric as a structural economic similarity measure}
\maketitle

\begin{abstract}
Structural similarity, the extent to which economic mechanisms respond alike to common shocks, matters for economic forecasting, policy transfer, and other decisions that rely on evidence from comparable settings. For mechanisms approximated by linear models, this idea has a simple geometric representation: The smaller the angle between their coefficient vectors, the more similarly they respond to the same shock. This paper introduces the cosine of this angle as a measure of structural similarity and develops the corresponding inference procedure. Applied to New Jersey's wage mechanism around the 1992 minimum-wage reform, the approach yields interesting evidence: Although the reform clearly affected the low-wage part of the wage distribution, the underlying mechanism linking worker characteristics to wages was no different from what would have prevailed without it.
\\
\\
\medskip
\noindent\textbf{Keywords:} Minimum wage, model comparison,  parameter heterogeneity \\
\noindent\textbf{JEL codes:} C12, C13, J31, J38
\end{abstract}

\section{Introduction}\label{sec1}
Economic units are inherently heterogeneous. Beyond differences in observed outcomes, they differ in the mechanisms through which economic shocks and policy interventions are translated into economic responses. Quantifying this structural diversity is essential for understanding heterogeneous adjustment processes, assessing the transferability of policies across settings, and constructing credible counterfactual predictions. Yet the analytical tools commonly used to study structural diversity remain limited. Conventional distance measures, including measures of economic distance based on trade costs \citep{Fisheretal2015}, compress multidimensional differences into scalar summaries and may therefore obscure economically relevant features of the underlying relationships. Clustering methods \citep{BonhommeManresa2015} take the opposite approach, but require many cross-sectional units to reliably detect latent groups, making them ill-suited to comparisons among a small number of economies. As a result, much of the nuanced heterogeneity that shapes economic behaviour remains analytically invisible.

The current research takes a different perspective on structural similarity by comparing the underlying economic mechanisms; that is, how economies respond to common shocks. To this end, it constructs response coefficient vectors for each economy and introduces a cosine-based measure of structural similarity defined by the angle between these vectors.

Structural similarity measurement opens a range of analytical possibilities across several domains. In regional economic development, it allows researchers to identify similarities and divergences in structural profiles and responses to economic shocks. This perspective is particularly valuable for investment and export strategies: actors can target economies whose reaction patterns diverge, thereby reducing correlated risks and revealing strategic opportunities that conventional indicators overlook. The structural similarity approach also provides a principled entry point into the longstanding challenge of assessing the external validity of policy interventions. Because policy effects depend heavily on institutional context, demographic structure, and temporal conditions, transferring results across settings is fraught with uncertainty. The measure of structural similarity provides an early signal of whether policy outcomes observed in one context are likely to generalise to another. Together, these contributions position the structural similarity measurement as a powerful, theoretically grounded instrument for advancing comparative economic research, improving policy evaluation, and enhancing strategic decision-making.

The current research aims to introduce a theoretically grounded and empirically validated framework for measuring structural economic similarity using a cosine-based model-similarity metric. The research relates to strands of econometric literature that account for unit heterogeneity by model estimation. In general, this class of methodologies captures only a few distinct degrees of diversity: full uniqueness of economic units (e.g., models with individual effects), groups of similar units (e.g., models with cluster-type heterogeneity), or complete identity of economic units. This leaves a clearly identifiable gap in the literature: the need to move beyond discrete classifications and analyse structural diversity along a continuous spectrum.

Moreover, existing measures of proximity are typically used only as classification tools. While they help identify groups of similar economic units or quantify heterogeneity, little attention has been paid to the economic interpretation of the distance itself. Consequently, these measures are rarely employed to determine whether empirical evidence can be generalised across economies or whether policies are likely to remain effective in different structural environments. In methods designed to identify cluster-type heterogeneity, the cluster structure is inferred from distances between vectors of estimated coefficients in the parameter space. The k-means approach of \cite{BonhommeManresa2015} is based on minimising squared Euclidean distances between individual-specific coefficient estimates and group-specific centroids. The CARDS method of  \cite{Keetal2015} identifies groups through ordered differences (L1 or L2) in estimated coefficients. More generally, methods for identifying grouped heterogeneity typically rely on Manhattan or Euclidean notions of distance. While these measures allow groups to be ranked according to their relative proximity, they provide little substantive interpretation beyond the clustering exercise itself.

A related strand of the literature focuses on testing for heterogeneity. Rather than estimating group structures directly, these approaches attempt to formalise the choice between strict parameter homogeneity, grouped heterogeneity, and unrestricted heterogeneity across units. Importantly, they reduce the problem to a sequence of binary decisions regarding the presence or absence of specific forms of heterogeneity. Prominent examples include the traditional F-test and the slope homogeneity test of \cite{PesaranYamagata2008}, which compares individual slope estimates with pooled estimators. \cite{PattonWeller2023} extend this perspective to grouped heterogeneity by testing whether cluster centres obtained from a k-means procedure are statistically distinct. Despite their methodological differences, all three approaches rely on test statistics based on quadratic-form distances, which can be interpreted as squared Mahalanobis distances.

Taken together, existing econometric approaches treat distance primarily as a computational device rather than an object of economic interest. While it facilitates classification, clustering, or the detection of heterogeneity, the information contained in the distance itself remains largely unused. This leaves the broader question of how structural similarity can be quantified, interpreted, and incorporated into economic analysis largely unanswered.

This paper addresses the problem by introducing the cosine metric as a measure of structural economic similarity between response coefficient vectors and by developing the corresponding framework for its identification, estimation, and statistical inference. The remainder of the paper is organised as follows. Section 2 develops the theoretical foundations for using the cosine metric to quantify structural similarity between economic mechanisms. Section 3 introduces the estimator, establishes its statistical properties, and develops the corresponding inference procedures. Its finite-sample performance is then assessed in Section 4 through Monte Carlo simulation. Section 5 illustrates the methodology using the New Jersey minimum wage reform to examine whether the policy altered the underlying response structure of the labour market beyond its conventional treatment effects. Section 6 discusses the limitations of the proposed methodology and outlines several directions for future research. Section 7 concludes.  

\section{Economic justification of the use of cosine metric}\label{sec2}
Economists usually illustrate a country's reaction to a shock using a small number of key macroeconomic parameters, such as GDP growth, inflation, unemployment, which are also widely used for cross-country comparisons. Yet, long-term economic resilience is rooted in economic mechanisms that can hardly be captured by comparing aggregate outcomes alone. Motivated by \cite{Sims1980}, economic mechanisms can be described through the way economic variables respond to shocks. In this spirit, consider an economic mechanism well approximated by the linear model \begin{equation}
y=X\beta+\epsilon,
\label{eq:linear model}
\end{equation}  where the coefficient vector $\beta$ characterises the response of the outcome variable to changes in the explanatory variables. Let $\beta_i$ and $\beta_j$ denote coefficient vectors corresponding to two economies, regions, sectors, or other economic units. The vectors of coefficients are vectors of marginal responses to unit changes in factor variables to be named as structural response vectors. Two economic units may exhibit the same aggregate change following a particular shock, $\Delta X_i\beta_i=\Delta X_j\beta_j$,
while their coefficient vectors differ. The same observed change can therefore be generated by fundamentally different underlying mechanisms—a distinction that a comparison of aggregate outcomes alone may obscure.

While direct comparison of coefficient vectors in  \eqref{eq:linear model} across economies does not, in general, permit conclusions about the equivalence of underlying economic mechanisms, it does provide information about the similarity in the direction of responses to common shocks or policy interventions. To compare the response vectors geometrically, we work with their normalised counterparts $\tilde\beta_i = {\beta_i}/{\|\beta_i\|},  \tilde\beta_j = {\beta_j}/{\|\beta_j\|}$ and construct an orthonormal basis for the plane spanned by them: the first basis vector $e_1 = \tilde\beta_j$ is the direction of economy $j$'s response vector, and the second, $e_2 = \tilde\beta_j^{\perp}$, is the unit vector in that plane orthogonal to $\tilde\beta_j$. In this basis, the normalised response structure of economy $i$ can be written as
\begin{equation}
\tilde\beta_i = \tilde\beta_j \cos\phi + \tilde\beta_j^{\perp}\sin\phi,
\label{eq:decomposition}
\end{equation}
where $\phi$ is the angle between the two structural response vectors. Thus, $\cos\phi$ measures the component of economy $i$'s response structure aligned with that of economy $j$, while $\sin\phi$ captures the orthogonal component reflecting their structural difference. A value of $\phi=0^{\circ}$ indicates that the coefficient vectors are proportional with a positive proportionality factor. In this case the economies have similar response structures but differ in the magnitude of their response to the same shock. The two economies, therefore, exhibit the same relative response pattern across factors: one economy may respond more strongly or more weakly to every factor, but the relative importance of the factors in generating the response remains unchanged. For $0^\circ<\phi<90^\circ$, the two economies exhibit partially similar response structures. They share some common features in the way changes in the underlying factors translate into changes in the outcome variable, but the relative importance of individual factors differs across economies. As the angle increases, the overall relative importance of the factors in generating the outcome response becomes increasingly different across the two economies, and the economies may also differ in the direction of the responses associated with particular factors. At $\phi=90^{\circ}$, the two economies exhibit fundamentally different ways of translating changes in the underlying factors into changes in the outcome variable. Finally, $\phi=180^{\circ}$ indicates that the coefficient vectors are proportional with a negative proportionality factor. The response patterns of the two economies are therefore opposite: factors associated with positive responses in one economy are associated with negative responses in the other, up to a common scale factor.

To express $\phi$ in terms of $\beta_i$ and $\beta_j$, take the dot product of both sides of \eqref{eq:decomposition} with $\tilde\beta_j$:
\[
\tilde\beta_i \cdot \tilde\beta_j = \cos\phi \,(\tilde\beta_j \cdot \tilde\beta_j) + \sin\phi \,(\tilde\beta_j^{\perp}\cdot \tilde\beta_j).
\]
Since $e_1, e_2$ form an orthonormal basis, $\tilde\beta_j\cdot\tilde\beta_j = \|\tilde\beta_j\|^2 = 1$ and $\tilde\beta_j^{\perp}\cdot\tilde\beta_j = 0$, so that
\[
\tilde\beta_i \cdot \tilde\beta_j = \cos\phi.
\]
Substituting $\tilde\beta_i = \beta_i/\|\beta_i\|$ and $\tilde\beta_j = \beta_j/\|\beta_j\|$ then gives the cosine metric
\begin{equation}
\theta_{ij} = \frac{\beta_i \cdot \beta_j}{\|\beta_i\|\,\|\beta_j\|}.
\label{eq:cosine metric}
\end{equation}
Thus $\phi$, and hence $\theta_{ij}$, is fully determined by the original (unnormalised) structural response vectors; the decomposition \eqref{eq:decomposition}, together with the derivation above, justifies the use of the cosine metric  $\theta_{ij}$ as a measure of structural similarity between economic units.

Revisiting the object of comparison stated at the beginning of this section, structural similarity is meant to capture how economies respond to shocks in the underlying factors. In a linear model, such responses are governed by the slope coefficients, whereas the intercept captures the baseline level of the outcome. This distinction carries directly into how the cosine metric is constructed. Calculated without the intercept, $\theta_{ij}$ isolates similarity in the underlying factor-response structure; calculated with the intercept included, by contrast, $\theta_{ij}$ mixes similarity in this response structure with similarity in baseline outcome levels, since both components enter the same coefficient vector and cannot be separately recovered from a single cosine value. Conflating these two objects into a single indicator obscures more than it reveals, since differences in outcome levels and differences in the structure governing responses to factor changes are conceptually distinct dimensions of economic similarity that need not move together. We therefore exclude the constant from the coefficient vector used to construct the cosine similarity metric in \eqref{eq:estimator} throughout the paper, treating this as our primary operationalisation of structural similarity.

Finally, a high cosine similarity does not by itself establish that two mechanisms are equivalent. The estimated coefficient vector is a linear approximation to a far richer, typically nonlinear aggregation of individual decision-making processes, and the metric can only compare mechanisms through the lens of this approximation. While $\theta_{ij}$ provides a rigorous, scale-invariant measure of the directional similarity of estimated responses, interpreting it as evidence of similarity in the underlying economic mechanisms themselves requires the more precise economic analysis.

\section{Estimating the cosine similarity}\label{sec3}
Consider an economic mechanism approximated by the linear model without constant
\begin{equation}
y = X\beta + \epsilon,
\label{eq:linear_model}
\end{equation}
and let $\beta_i, \beta_j \in \mathbb{R}^k$ denote the coefficient vectors describing this mechanism in economies $i$ and $j$. The cosine metric measuring the directional similarity between these two structural response vectors is
\begin{equation}
\theta_{ij} = \theta(\beta_i,\beta_j) = \frac{\beta_i \cdot \beta_j}{\|\beta_i\|\,\|\beta_j\|}.
\label{eq:cosine}
\end{equation}

\begin{assumption}[Identification]
\label{as:identification}
The coefficient vectors $\beta_i$ and $\beta_j$ are point identified.
\end{assumption}

\begin{assumption}[Nondegeneracy]
\label{as:nondegenerate}
$\|\beta_i\| > 0$ and $\|\beta_j\| > 0$.
\end{assumption}

\begin{assumption}[Consistency]
\label{as:consistency}
$\hat\beta_i \xrightarrow{p} \beta_i$ and $\hat\beta_j \xrightarrow{p} \beta_j$ as $n \to \infty$.
\end{assumption}

\begin{assumption}[Joint asymptotic normality]
\label{as:normality}
Let $\gamma = (\beta_i', \beta_j')' \in \mathbb{R}^{2k}$ and $\hat\gamma = (\hat\beta_i', \hat\beta_j')'$. Then
\begin{equation}
\sqrt{n}(\hat\gamma - \gamma) \xrightarrow{d} N(0, \Sigma),
\label{eq:clt}
\end{equation}
where $\Sigma$ is the $2k \times 2k$ asymptotic variance-covariance matrix of $\hat\gamma$.
\end{assumption}

\begin{remark}
When $\hat\beta_i$ and $\hat\beta_j$ are estimated on independent samples (the typical case when $i$ and $j$ index distinct economies), $\Sigma$ is block-diagonal, $\Sigma = \mathrm{diag}(\Sigma_i, \Sigma_j)$, which simplifies the variance expressions below.
\end{remark}

\begin{theorem}[Identification of $\theta_{ij}$]
\label{th:identification}
Under Assumptions \ref{as:identification}--\ref{as:nondegenerate}, the cosine similarity $\theta_{ij}$ is point identified.
\end{theorem}
\begin{proof}
By Assumption \ref{as:identification}, $\beta_i$ and $\beta_j$ are uniquely determined by the population distribution of $(y,X)$ in each economy. By Assumption \ref{as:nondegenerate}, $\theta(\cdot,\cdot)$ is a well-defined, single-valued function on $(\mathbb{R}^k\setminus\{0\})^2$. A deterministic, well-defined function of uniquely determined arguments is itself uniquely determined; hence $\theta_{ij}=\theta(\beta_i,\beta_j)$ is point identified.
\end{proof}

\begin{definition}[Estimator]
The cosine metric is estimated by its sample analogue,
\begin{equation}
\hat{\theta}_{ij} = \theta(\hat\beta_i,\hat\beta_j) = \frac{\hat{\beta}_i \cdot \hat{\beta}_j}{\|\hat{\beta}_i\|\,\|\hat{\beta}_j\|}.
\label{eq:estimator}
\end{equation}
\end{definition}

\begin{theorem}[Consistency]
\label{th:consistency}
Under Assumptions \ref{as:nondegenerate} and \ref{as:consistency}, $\hat\theta_{ij} \xrightarrow{p} \theta_{ij}$.
\end{theorem}
\begin{proof}
By Assumption \ref{as:nondegenerate}, $\theta(\cdot,\cdot)$ is continuous at $(\beta_i,\beta_j)$. By Assumption \ref{as:consistency}, $\hat\beta_i\xrightarrow{p}\beta_i$ and $\hat\beta_j\xrightarrow{p}\beta_j$. The continuous mapping theorem then gives $\hat\theta_{ij}=\theta(\hat\beta_i,\hat\beta_j)\xrightarrow{p}\theta(\beta_i,\beta_j)=\theta_{ij}$.
\end{proof}

\begin{theorem}[Asymptotic normality]
\label{th:asym_normal}
Under Assumptions \ref{as:nondegenerate} and \ref{as:normality},
\begin{equation}
\sqrt{n}\left(\hat\theta_{ij} - \theta_{ij}\right) \xrightarrow{d} N\left(0,\ V_{\hat\theta}\right), \qquad V_{\hat\theta} = \nabla\theta' \, \Sigma \, \nabla\theta,
\label{eq:delta}
\end{equation}
where $\nabla\theta = \left(\partial\theta/\partial\beta_i',\ \partial\theta/\partial\beta_j'\right)' \in \mathbb{R}^{2k}$ is evaluated at $(\beta_i,\beta_j)$.
\end{theorem}
\begin{proof}
By Assumption \ref{as:nondegenerate}, $\theta(\gamma)$ is continuously differentiable at $\gamma=(\beta_i',\beta_j')'$. Applying the multivariate delta method to the joint CLT in Assumption \ref{as:normality} (equation \eqref{eq:clt}) yields \eqref{eq:delta}.
\end{proof}

\begin{corollary}[Closed-form gradient]
\label{cor:gradient}
Writing $\tilde\beta_i = \beta_i/\|\beta_i\|$ and $\tilde\beta_j = \beta_j/\|\beta_j\|$,
\begin{equation}
\frac{\partial \theta}{\partial \beta_i} = \frac{1}{\|\beta_i\|}\left(\tilde\beta_j - \theta_{ij}\,\tilde\beta_i\right), \qquad
\frac{\partial \theta}{\partial \beta_j} = \frac{1}{\|\beta_j\|}\left(\tilde\beta_i - \theta_{ij}\,\tilde\beta_j\right).
\label{eq:gradient}
\end{equation}
\end{corollary}

\begin{corollary}[Feasible inference]
\label{cor:feasible}
Let $\hat\Sigma$ be a consistent estimator of $\Sigma$ and $\hat\nabla\theta$ the gradient in \eqref{eq:gradient} evaluated at $(\hat\beta_i,\hat\beta_j)$. Then $\hat V_{\hat\theta} = \hat\nabla\theta'\hat\Sigma\hat\nabla\theta \xrightarrow{p} V_{\hat\theta}$, and by Slutsky's theorem,
\begin{equation}
\frac{\hat\theta_{ij}-\theta_{ij}}{\sqrt{\hat V_{\hat\theta}/n}} \xrightarrow{d} N(0,1),
\end{equation}
yielding the asymptotically valid $(1-\alpha)$ confidence interval $\hat\theta_{ij} \pm z_{1-\alpha/2}\sqrt{\hat V_{\hat\theta}/n}$.
\end{corollary}

\section{Simulation study}

We next examine the finite-sample performance of the cosine similarity estimator and the associated variance estimator derived in the previous section. Specifically, the simulations address two questions: (i) how accurately does the estimator in \eqref{eq:estimator} recover the true cosine similarity between two economic mechanisms, and (ii) how accurately does the variance estimator in Corollary~\ref{cor:feasible} approximate the sampling variance of the cosine estimator?

To parallel the empirical application, we formulate the baseline DGP as a three-covariate linear mechanism, analogous to the Mincer earnings equation considered in the empirical section:
\begin{equation*}
Y_{it}
=
\beta_{1,it}X_{1,it}
+\beta_{2,it}X_{2,it}
+\beta_{3,it}X_{3,it}
+\varepsilon_{it},
\end{equation*}
where
\[
X_{1,it}, X_{2,it}, X_{3,it}
\overset{\mathrm{i.i.d.}}{\sim}
\mathcal{N}(0,1),
\]
and $\varepsilon_{it}$ is an independently distributed disturbance. We consider two economic units governed by comparable economic mechanisms of this form; their coefficient vectors capture the responses to unit changes in a common set of covariates, and the estimator in \ref{eq:estimator} measures the similarity in the direction of the two responses. 

Since economic mechanisms are estimated from populations of very different sizes (from a handful of narrowly defined economic units to broad aggregates spanning many agents, as is typical of individual-level applications) we examine the estimator's performance across sample sizes $n \in \{50, 200, 500, 1000\}$. We consider true cosine similarities of $\theta_0 \in \{0, 0.5, 0.9, 0.99\}$, representing substantially different degrees of similarity between the underlying mechanisms. Each design is evaluated using $N_{\mathrm{MC}} = 10{,}000$ Monte Carlo replications.

\begin{figure}[!h]
\begin{center}
\includegraphics[width=0.8\textwidth]{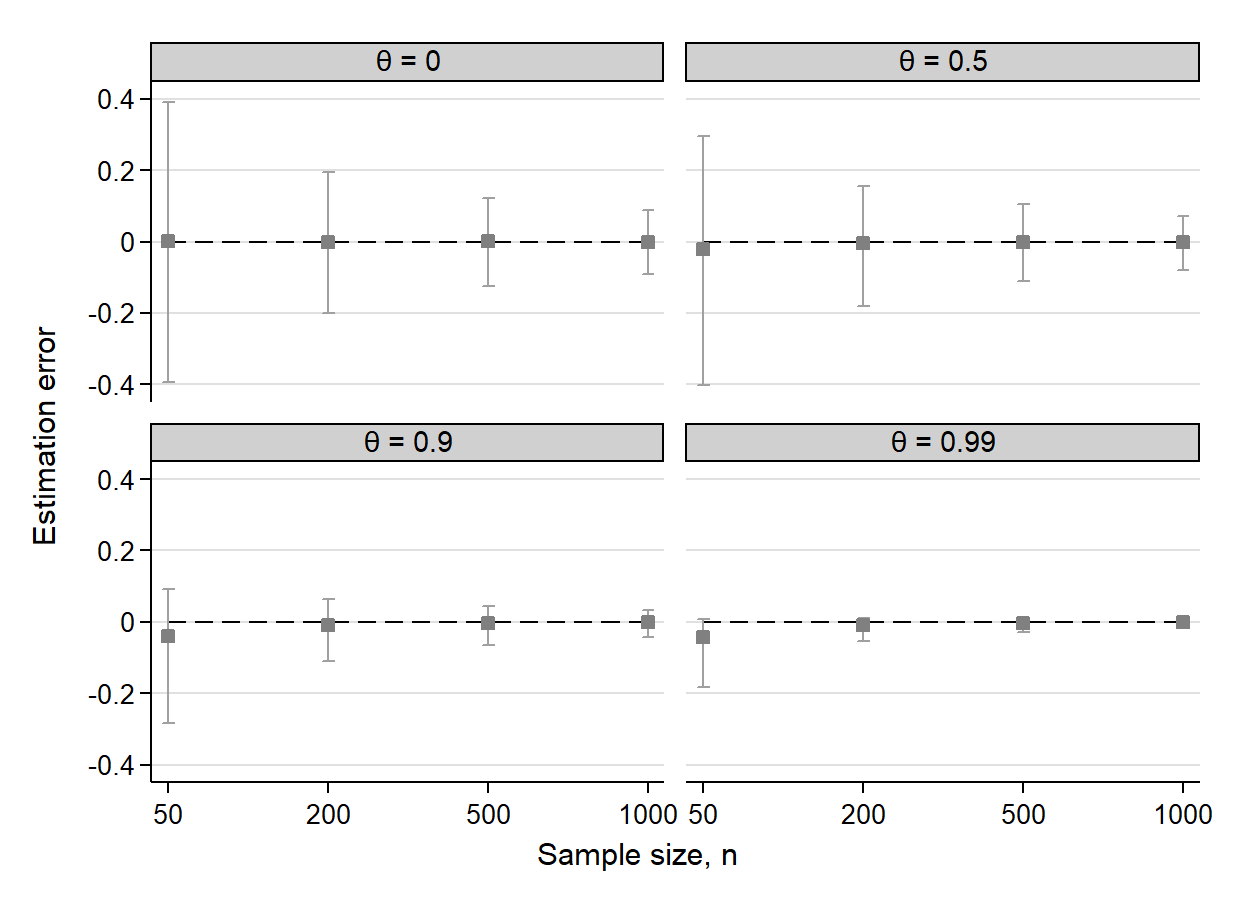}
\caption{Finite-sample performance of the cosine estimator}
\label{fig:MC}
\end{center}
\footnotesize{\textit{Note:} The figure reports the mean estimation error, $\hat{\theta}-\theta$, across Monte Carlo replications, with the endpoints of the vertical lines marking the 2.5th and 97.5th percentiles of the distribution of the estimation error.}
\end{figure}

The results shown in Figure~\ref{fig:MC} support the consistency of the cosine estimator established in Theorem~\ref{th:consistency}: across all values of the true cosine similarity, the mean estimation error approaches zero and the dispersion of the estimator decreases as the sample size increases. For $\theta=0$ and $\theta=0.5$, the estimator is nearly unbiased even in relatively small samples. At higher levels of $\theta$, a more pronounced finite-sample downward bias appears, which declines rapidly with the sample size. The delta-method variance estimator in Corollary~\ref{cor:feasible} performs comparably well: across the interior designs, the average standard error based on the variance estimator closely matches the empirical Monte Carlo standard deviation (see Appendix Table~\ref{tab:mc_results}).

A distinct pattern emerges when the true coefficient vectors are nearly collinear ($\theta=0.99$ in the simulations). In this near-boundary case, point estimation improves rapidly with the sample size: both the bias and
the dispersion of the cosine estimator become very small already at $n=500$. This improvement in point estimation, however, is not accompanied by a corresponding improvement in inference. The coverage probability of the nominal 95\% confidence interval is conservative in the smallest
sample and falls somewhat below its nominal level in the larger samples (see Appendix Table~\ref{tab:mc_results}). This pattern differs from the interior designs, where coverage remains close to the nominal level as estimation uncertainty decreases. The near-boundary behaviour is consistent with the fact that cosine similarity is bounded above by one and becomes locally insensitive to perturbations of the coefficient vectors as their directions approach collinearity. Consequently, the first-order delta-method approximation is less accurate when the population cosine similarity is close to its boundary. In applications where formal testing of exact similarity is required, rather than directly testing the boundary hypothesis $H_0:\theta=1$, where the first-order delta-method approximation becomes locally
degenerate, it may be preferable to test the proportionality of the underlying response vectors using a Wald-type test. This formulation characterises the same notion of perfect directional similarity while avoiding inference on the cosine metric directly at its boundary.

Importantly, however, exact equality need not be the economically relevant question. A cosine similarity approaching unity may itself provide strong evidence that the underlying mechanisms are closely aligned. In such cases, the relevant question shifts from whether the mechanisms are statistically indistinguishable to whether the remaining differences are economically meaningful. This calls for a substantive assessment of those differences
in the context of the decision at hand, such as policy transfer, benchmarking, or the choice of an appropriate comparison unit.

\section{Empirical Analysis: The New Jersey Minimum Wage Reform}

The 1992 increase in New Jersey's minimum wage, from \$4.25 to \$5.05 per hour, remains one of the most widely studied episodes in labour economics. Since \citet{CardKrueger1994} found no evidence of an employment decline in their comparison of fast-food restaurants in New Jersey
and neighbouring Pennsylvania, a large body of subsequent research has examined the resulting changes in employment and wages using different data and identification strategies \citeg{Dubeetal2010,Cengizetal2019, Wied2026}.

Yet this literature leaves a different question largely unexplored: did the reform reshape the underlying structure of the labour market, not merely its outcomes? Conventional treatment-effect analysis asks whether an intervention moves wages or employment. It says less about whether the relationships generating those outcomes change. A policy may therefore leave conventional outcome measures largely unchanged while still altering how worker characteristics relate to labour-market outcomes, through several plausible channels. A higher wage floor, for instance, may change firms' hiring and retention decisions, increasing the relative importance of observable productivity signals (education, labour-market experience) when firms select among workers whose wages are now compressed by the new minimum. The reform can also send spillovers further up the wage distribution: firms adjusting pay differentials to preserve internal hierarchies would shift the returns to worker characteristics even for employees who initially earned above the statutory minimum. Employment composition might shift too, as lower-productivity workers exit and firms substitute toward more experienced or more educated hires. Higher relative wages in New Jersey could, in turn, affect labour supply (participation, job search, and possibly cross-border commuting or migration), reshaping the pool of workers competing for jobs. Together, these channels suggest that the reform could reshape how wages relate to worker characteristics, independent of any shift in their average level. This is a type of change conventional treatment-effect estimates are not designed to detect. Structural similarity analysis offers a complementary lens, asking not whether the policy moved outcomes, but whether it left the wage-setting mechanism itself intact.

\subsection{New Jersey's Wage Mechanism Around the Reform}

Consider a wage mechanism described by the Mincer equation \citep{Mincer1974}. Grounded in human capital theory \citep{Becker1964} and the life-cycle model of human capital accumulation, the Mincer framework captures essential features of individuals' schooling and labour-market decisions and relates them to the evolution of earnings over the working life. Although there is evidence that some of the key assumptions underlying the conventional Mincer specification—in particular, the separability of schooling and labour-market experience, and the log-linearity of earnings in schooling—have become less consistent with observed wage patterns since the 1980s \citeg{Heckmanetal2003,Lemieux2006}, the Mincer equation remains a widely used and economically interpretable benchmark for characterising the relationship between wages, schooling, and labour-market experience \citep{Lemieux2006}.

Using the Mincer equation to characterise the wage mechanism, we ask two related questions: did New Jersey's wage mechanism change following the April 1992 minimum-wage increase, and did its similarity to the wage mechanisms of other states change as a result?

We examine these questions using the CPS Basic Monthly files from January 1989 through December 1993. The sample is restricted to workers with valid information on hourly wages, education, and labour-market experience. The estimation sample covers 51 states; the median state contributes approximately 53,000 observations before the reform and 28,000 after the reform. Following the Mincer specification, the outcome is the log of nominal hourly wages, and the wage mechanism is characterised by years of education, potential labour-market experience, and experience squared. Potential experience is defined as age minus years of education minus six. To place the coefficients on a common scale and facilitate comparisons across states and periods, the covariates are standardised using the pooled sample. Education and potential experience are transformed to have zero mean and unit variance, and the quadratic term is constructed as the square of standardised experience. A constant is included in every specification, but it is excluded from the cosine-metric calculation.

We estimate the Mincer equation by GMM and allow the covariance structure to accommodate two sources of dependence inherent in the data. The rotating-panel structure of the CPS means that observations on the same individual may be correlated across survey months; a state-month cell contains, on average, approximately 3,000 distinct individuals. At the same time, common macroeconomic shocks may generate contemporaneous dependence across states within a given month. We therefore construct the variance–covariance matrix of the coefficient estimates allowing for dependence along both the individual and year–month dimensions. The resulting coefficient estimates are substituted into \eqref{eq:estimator} to obtain the corresponding cosine similarities between the estimated coefficient vectors.

We first use these similarities to assess how stable state wage mechanisms were around the reform. For each state, we compare the coefficient vector estimated before April 1992 with its post-reform counterpart. A cosine similarity close to one indicates little change in the relative configuration of the estimated coefficients, while lower values indicate a larger change in the wage mechanism. New Jersey's similarity alone, however, does not reveal whether its change was unusual. We therefore locate New Jersey's pre--post similarity within the cross-state distribution of the same measure, using changes in other states over the same period as a benchmark.

\begin{figure}[!h]
\centering
\includegraphics[width=0.8\textwidth]{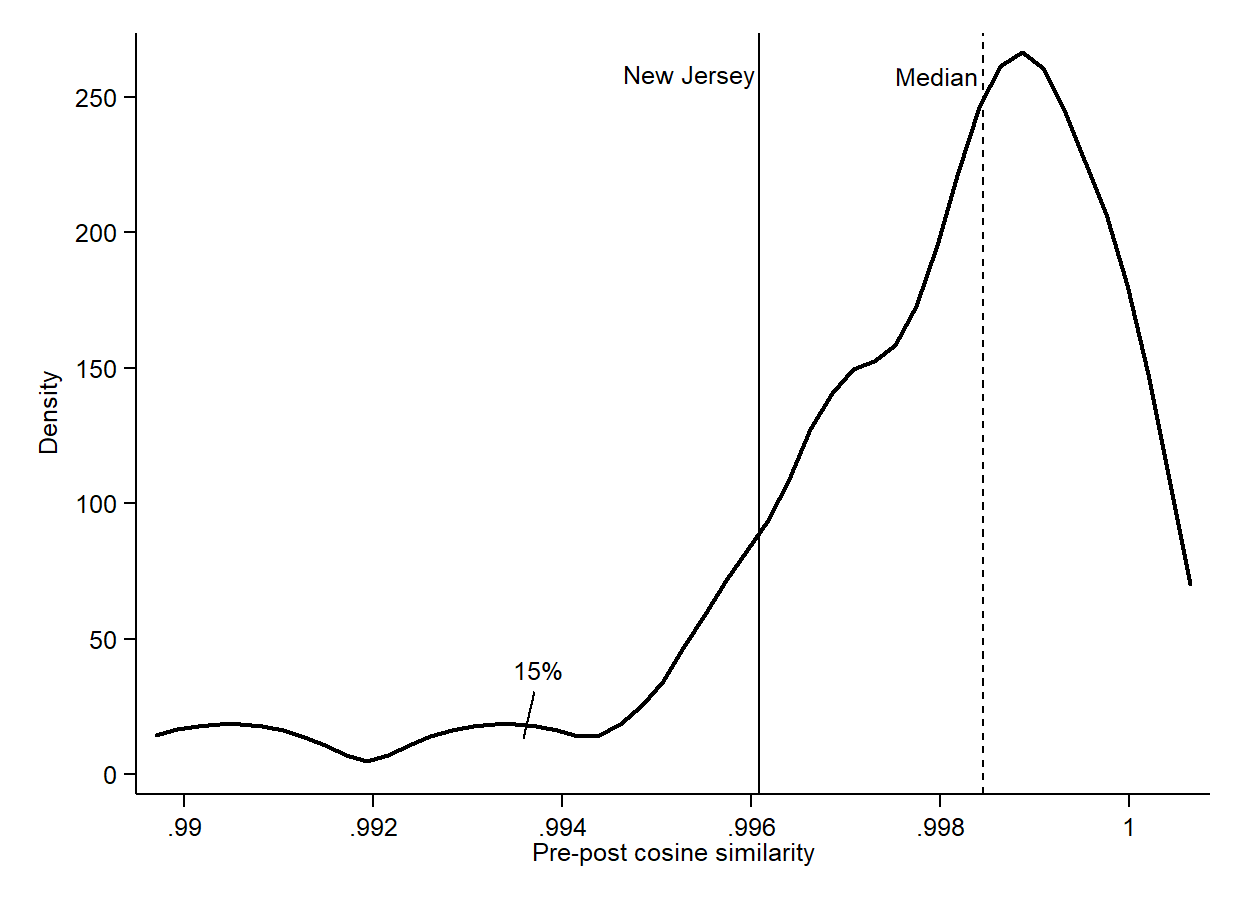}
\caption{Distribution of within-state wage-mechanism similarity around the New Jersey reform}
\label{fig:Distribution}
\end{figure}

Figure \ref{fig:Distribution} reveals a striking concentration of within-state similarities near one. For most states, the estimated relationship between wages, education, and experience therefore changed relatively little over the period surrounding the New Jersey reform. New Jersey departs from this general pattern: its pre--post similarity lies toward the lower tail of the distribution and well below the cross-state median. Put differently, the reconfiguration of New Jersey's wage mechanism was larger than that observed in most other states over the same period. This conclusion is reinforced by a joint Wald test, which rejects equality of New Jersey's pre- and post-reform coefficient vectors (p=0.009). Thus, both the cosine-based comparison and a conventional test of coefficient stability point to a statistically detectable change in New Jersey's wage mechanism. This comparison does not, on its own, identify the minimum-wage increase as the source of the change. It does show, however, that New Jersey experienced an unusually pronounced reconfiguration relative to the contemporaneous changes observed across other states.

A large within-state change, however, does not by itself reveal whether New Jersey grew more or less similar to other states: if other states shifted in the same direction, their similarity to New Jersey could remain essentially unchanged, while even a modest change in New Jersey alone could reshape its relative standing with particular states. This distinction motivates the comparison in Figure \ref{fig:Scatter}, which plots each state's cosine similarity with New Jersey after the reform against its similarity before the reform. The 45-degree line represents unchanged relative similarity: observations above it indicate convergence toward New Jersey's wage mechanism, while observations below it indicate divergence.

\begin{figure}[!h]
\begin{center}
\includegraphics[width=0.9\textwidth]{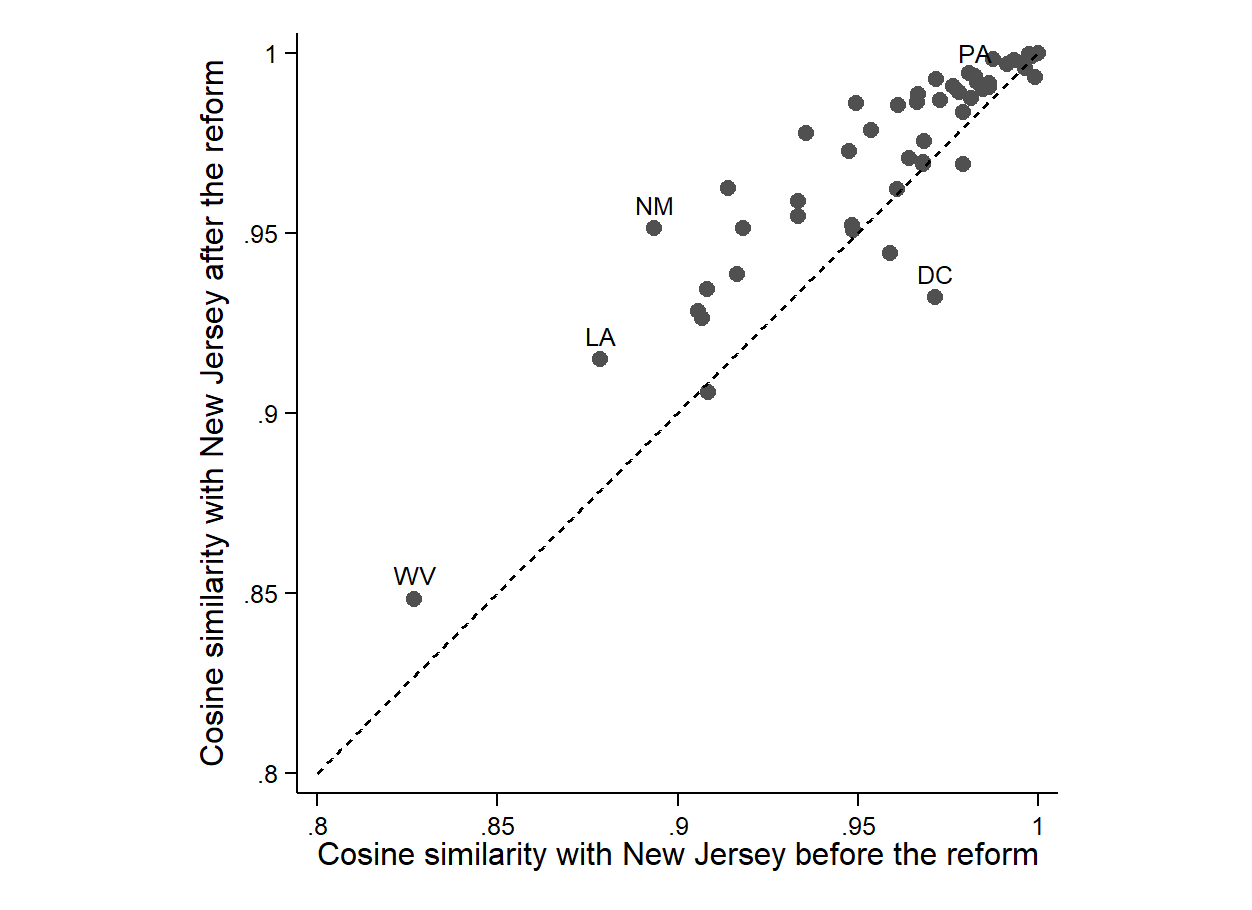}
\caption{Similarity of state wage mechanisms to New Jersey before and after the reform} 
\label{fig:Scatter}
\end{center}
\footnotesize{\textit{Note:} Selected labels highlight states with particularly large changes in similarity, together with Pennsylvania as a neighboring benchmark: WV = West Virginia; LA = Louisiana; NM = New Mexico; DC = District of Columbia; PA = Pennsylvania}
\end{figure}

The pattern that emerges is strikingly consistent: most states became closer to New Jersey after the reform. West Virginia and Louisiana, despite comparatively low pre-reform similarity, both converge toward New Jersey, while New Mexico shows the largest movement in the sample, rising from roughly 0.90 before the reform to approximately 0.95 afterward. States that were already highly similar to New Jersey before the reform, with cosine similarity near unity, also move upward, if only modestly, clustering just above the 45-degree line. Against this near-uniform convergence, the District of Columbia stands out as the one clear exception, falling from a pre-reform similarity of 0.97 to 0.93 — the only state to become measurably less similar to New Jersey after the reform. The reform, in other words, coincided with a broad-based convergence of state wage mechanisms toward New Jersey's, largely independent of each state's starting position, with the District of Columbia as a singular, conspicuous departure from that pattern.

Figure \ref{fig:Scatter} describes how New Jersey's relative position changed across the full set of states, but these movements must also be viewed in light of their estimation uncertainty. Figure \ref{fig:CIs} therefore reports the change in cosine similarity, together with 95\% confidence intervals, for five informative comparison states \protect\footnotemark. Pennsylvania and New York provide natural geographic benchmarks, while three additional states are selected to represent low, median, and high pre-reform similarity to New Jersey. The resulting comparison spans markedly different initial positions without conditioning the selection on the estimated post-reform change.  
\footnotetext{
Notes: Confidence intervals widen with lower coefficient-vector norms and with greater angular distance between mechanisms — both channels through which the cosine metric's sensitivity to estimation error increases, as shown by its gradient, $\sin\phi/||\beta||$
(Corollary~\ref{cor:gradient}). Since 
$\phi$ increases as similarity to New Jersey falls, confidence intervals widen accordingly. Reported intervals use the full delta-method variance of the estimated change in cosine similarity.
}

\begin{figure}[!h]
\begin{center}
\includegraphics[width=0.9\textwidth]{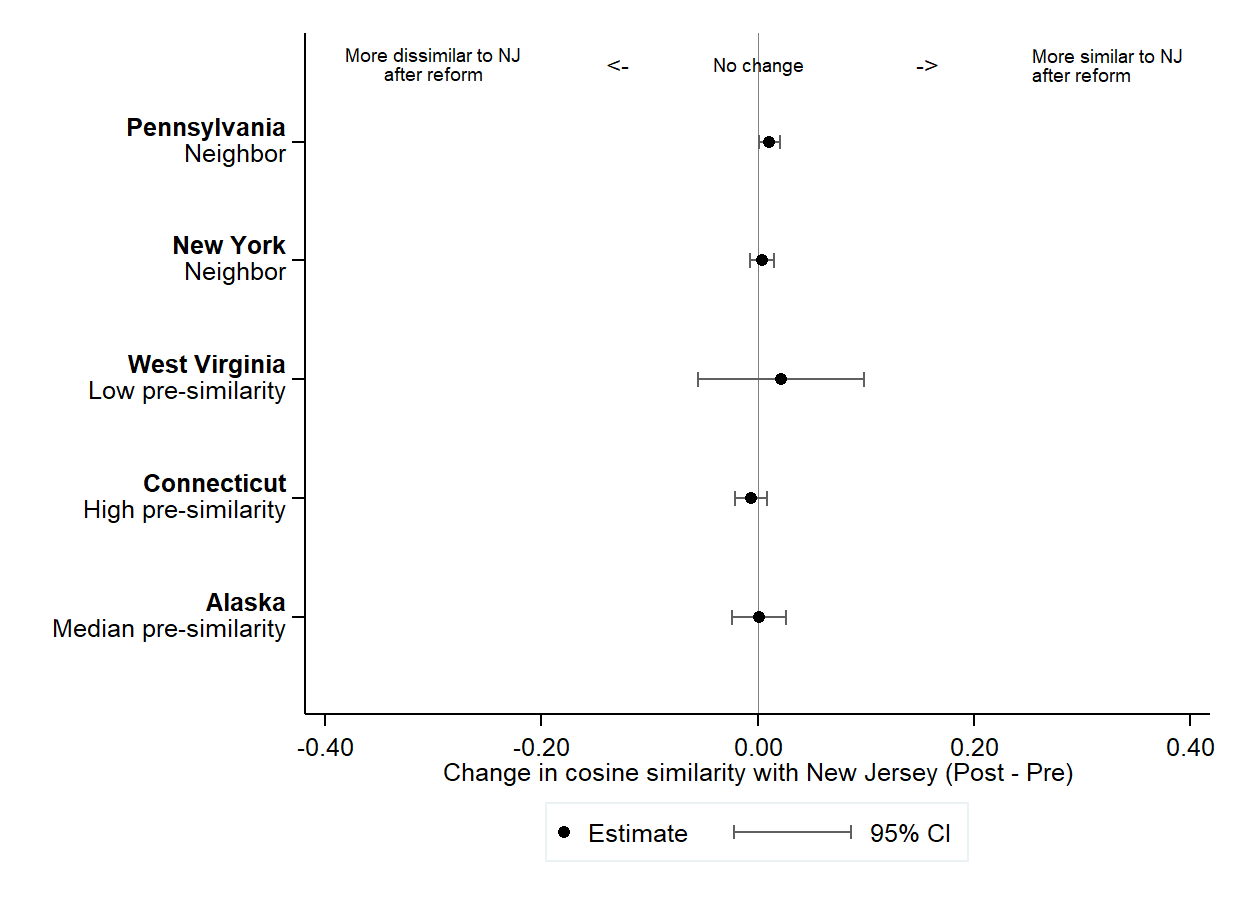}
\caption{Estimated Changes in Wage-Mechanism Similarity to New Jersey}
\label{fig:CIs}
\end{center}
\footnotesize{\textit{Note:}  Confidence intervals are constructed from the joint delta-method variance of the pre- and post-reform cosine estimates, allowing for dependence along both the individual and year–month dimensions.}
\end{figure}

The estimates point to a notable stability in New Jersey's relative position despite the substantial change in its own wage mechanism documented above. Similarity with the neighbouring states and with the state most similar to New Jersey before the reform changes little across the two periods; the median- and low-similarity benchmarks likewise remain centred near zero, showing no common movement toward or away from New Jersey. The reform, in other words, reconfigured New Jersey's wage-determination structure internally without reordering its position among other states' wage mechanisms.

\subsection{Beyond Common Trends: A Reform-Specific Effect in New Jersey's Wage Mechanism}
\label{subsec:beyond_common_trends}
The preceding analysis establishes that New Jersey experienced a statistically significant change in its wage-mechanism coefficient vector following the reform. This pre--post comparison, however, does not distinguish a reform-specific change from broader structural changes affecting wage determination across states over the same period. To separate these two components, we complement the analysis with a synthetic-control comparison following the approach of \citet{Wied2026}. The key insight of \citet{Wied2026} is that, under a parallel-trends condition in the parameter space, a weighted combination of parameter changes in untreated units can be used to recover the counterfactual evolution of the treated unit's parameters in the absence of treatment. We exploit this insight to construct a synthetic counterfactual for New Jersey's wage-mechanism coefficient vector from those estimated for untreated states. Specifically, the synthetic counterfactual is constructed as a weighted combination of the untreated states' post-treatment coefficient vectors, using the weight estimates from \citet{Wied2026}. These weights are chosen to provide the best approximation of New Jersey's pre-treatment coefficient path by the corresponding pre-treatment coefficient paths of the untreated states in the distribution-regression model. Since \citet{Wied2026} conducts his empirical analysis using the same dataset and model specification as in our application, this provides a basis for applying the estimated weights in our setting. As a simple benchmark, we also construct the counterfactual coefficient vector by assigning equal weights to all donor states. Unlike the Wied-based synthetic counterfactual, this specification does not exploit pre-treatment fit and is therefore not intended as an alternative estimator, but provides a useful reference for assessing the role of the estimated synthetic-control weights. We then compare New Jersey's observed post-reform coefficient vector with its synthetic counterpart using the cosine similarity measure. This comparison allows us to assess whether the post-reform reconfiguration of New Jersey's wage mechanism exceeds the contemporaneous structural changes captured by the untreated states.

\begin{table}[htbp]
\centering
\caption{Structural Similarity of New Jersey and Synthetic New Jersey}
\label{tab:cosine_counterfactual}
\begin{tabular}{lccc}
\toprule

    & Cosine similarity 
    & Standard error 
    & 95\% CI \\
\midrule
Synthetic counterfactual
    & 0.996 
    & 0.005
    & [0.986, 1.000] \\

Naive equal-weight benchmark
    & 0.984 
    & 0.004 
    & [0.976, 0.993] \\
\bottomrule
\end{tabular}

\vskip 0.2cm
\begin{minipage}{\textwidth}
\footnotesize
\textit{Notes:} The synthetic counterfactual uses weight estimates from
\citet{Wied2026}. Standard errors are obtained using the delta method. States are assumed to be independent to keep estimation tractable, as allowing cross-state dependence increases the dimension of the covariance problem by a factor of S, where S is the number of states.
\end{minipage}
\end{table}

The reported results (Table~\ref{tab:cosine_counterfactual}) show a cosine similarity of 0.996 between New Jersey's post-treatment wage mechanism and its synthetic counterfactual, and the null hypothesis of perfect directional similarity, $H_0:\theta=1$, is not rejected, in contrast tom the case of the equal-weight benchmark. This result warrants caution, however: as cosine similarity approaches unity, the metric's gradient vanishes, the delta-method variance degenerates, and inference based directly on $\theta$ becomes unreliable near this boundary. We therefore complement it with a Wald-type test of the restriction $\boldsymbol\beta_{NJ}=\lambda\boldsymbol\beta_{CF}$, $\lambda>0$, applied directly in the coefficient space \protect\footnotemark. For the synthetic counterfactual, this restriction holds ($\chi^2(2)=3.58, p=0.167$), consistent with the direct cosine-based test: a synthetic counterfactual matched to New Jersey's pre-reform trajectory shows no detectable deviation from its post-treatment mechanism.

\footnotetext{
The proportionality coefficient $\lambda$ is estimated as the coefficient of the orthogonal projection of $\boldsymbol{\beta}_{\mathrm{NJ}}$ onto the one-dimensional subspace spanned by $\boldsymbol{\beta}_{\mathrm{CF}}$ by solving $\hat{\lambda} = \arg\min_{\lambda} \left\| \boldsymbol{\beta}_{\mathrm{NJ}} - \lambda\boldsymbol{\beta}_{\mathrm{CF}} \right\|^2$.
The covariance matrix estimator for the restriction $\boldsymbol{\beta}_{\mathrm{NJ}}=\lambda\boldsymbol{\beta}_{\mathrm{CF}}$
 is generically full rank in finite samples, yet under the null hypothesis it collapses to an asymptotic covariance matrix whose rank is one less than the dimension of the coefficient vector. As a result, the conventional Wald statistic does not have its standard chi-square limiting distribution, leading to size distortions in hypothesis testing \citep{DufourEtAl2025}. To satisfy the rank condition of \citet{Andrews1987}, under which a Wald-type statistic based on the Moore–Penrose generalized inverse of the covariance matrix estimator is asymptotically $\chi^2$
-distributed provided the estimator's rank converges to the asymptotic rank, we construct a null-imposed covariance matrix estimator whose rank matches the asymptotic rank of the covariance matrix under the null. The resulting modified Wald statistic is asymptotically distributed as $\chi^2(3)$.
}

This conclusion is strengthened by a further check: robustness to the inclusion of the intercept. The exclusion of the constant from our primary similarity measure was motivated by the premise that a mechanism's level and its structure are conceptually distinct margins of adjustment that need not move together, and indeed did not in the within-state comparison reported earlier. Here, by contrast, the two specifications agree. Including the constant, the cosine similarity between New Jersey's observed and synthetic counterfactual coefficient vectors is 0.973, with the corresponding proportionality restriction not rejected ($\chi^2(3)=5.41$, $p=0.144$); excluding it, the cosine similarity rises to 0.996, again with no evidence against proportionality. Since the constant, under our standardised specification, represents the log wage of an individual with average schooling and experience, this agreement indicates that neither the average wage level nor the relative importance of worker characteristics differs detectably between New Jersey's post-treatment mechanism and its properly weighted counterfactual. Taken together, these results recast the change documented in the within-state comparison above: much of the structural change in New Jersey`s wage mechanism appears to reflect trends shared with other states over the same period, rather than an adjustment specific to the reform itself.

A large literature documents that minimum-wage increases affect the low end of the wage distribution \citep{CardKrueger1994,Dubeetal2010,Cengizetal2019}, with more recent evidence pointing to heterogeneous effects concentrated among young, low-educated workers \citep{Wied2026}. Our results extend this picture along a different dimension: Despite the reform’s clear effects on wages, we find no detectable reform-specific change in the underlying mechanism linking worker characteristics to wages.

\section{Discussion}
\label{sec:discussion}

This paper studies structural similarity between economic mechanisms, defined in terms of their responses to common shocks. Applying the cosine metric to this problem requires translating a purely mathematical concept into an economically meaningful measure. Although the cosine metric is mathematically well-defined for any pair of coefficient vectors, its economic interpretation is valid only when the vectors represent comparable economic mechanisms. In this context, comparability refers to mechanisms describing the decisions of the same type of economic agents, operating at the same level of aggregation, and represented within a common modelling framework.

In practice, a lack of comparability between economic mechanisms does not necessarily preclude structural similarity analysis. Consider two economies with fundamentally different supply-side structures. In the first, production decisions are made primarily by decentralised firms operating in competitive or oligopolistic markets. In the second, production is heavily influenced by government intervention through production quotas, price controls, or direct ownership of firms. Although both settings involve a supply-side mechanism, the estimated coefficients summarise fundamentally different behavioural processes and therefore cannot be interpreted as describing the same economic mechanism.

Nevertheless, the lack of comparability may concern only part of the overall economic system. While the supply side differs substantially across the two economies, consumer demand may still be governed by similar behavioural principles. Structural similarity can then be evaluated separately for the demand equations, providing a meaningful comparison of the comparable component of the economic system. Extending the proposed framework to compare individual blocks of larger structural models represents a promising direction for future research.

A different challenge arises when the chosen econometric model fails to adequately represent the underlying economic mechanism. Since the cosine measure compares estimated response coefficient vectors rather than the mechanisms themselves, its validity depends on whether the estimated model captures the essential structural relationships. This issue is particularly relevant when inherently nonlinear mechanisms are approximated by linear specifications.

Consider, for example, a production technology exhibiting threshold effects, multiple production regimes, or strong complementarities between capital and labour. While a linear model may provide a reasonable approximation over a limited range of observations, it may fail to capture the underlying mechanism more generally. Similarly, inflation dynamics may respond approximately linearly to small supply shocks but exhibit pronounced nonlinearities following large disturbances. In such cases, differences between the estimated coefficient vectors may reflect model misspecification rather than genuine structural differences.

The proposed framework is therefore intended to compare response structures represented within a common and adequately specified econometric model. Extending the methodology to nonlinear specifications and more flexible structural representations constitutes a natural direction for future research.

Taken together, these considerations highlight that structural similarity is a multifaceted economic concept rather than a purely mathematical one. Although the cosine metric provides a rigorous statistical measure of similarity between response coefficient vectors, its economic interpretation depends on the nature of the mechanisms being compared, the validity of their econometric representation, and the context of the empirical application. Recognising these dimensions not only clarifies the scope of the proposed methodology but also identifies several promising directions for its further development.

\section{Conclusion}

This paper develops the concept of structural similarity analysis, which shifts the object of comparison from observed economic outcomes to the mechanisms that generate them. Rather than asking whether countries or regions are similar in terms of macroeconomic indicators, the proposed approach asks whether the relationships governing the decisions and outcomes of economic agents are themselves similar. This perspective broadens the scope of comparative analysis and provides a potential basis for assessing the external validity of policy evidence: economies governed by similar mechanisms may provide more informative benchmarks for anticipating responses to common shocks or policy interventions.

We operationalise this idea for economic mechanisms represented by linear models. We use cosine similarity to measure the proximity of coefficient vectors, propose an estimator of structural similarity, and develop inference for the resulting similarity measure. The comparison is meaningful for mechanisms that are comparable in economic content—that is, mechanisms describing the behavior of the same type of economic agents, at the same level of aggregation, and within a common modelling framework. Within this class, the approach provides a parsimonious way to quantify both the stability of a mechanism over time and its similarity across economic environments.

The empirical application to the 1992 New Jersey minimum-wage reform illustrates what this perspective can reveal beyond conventional treatment-effect analysis. Three features stand out. First, New Jersey experienced a comparatively large change in its estimated wage mechanism relative to contemporaneous changes in other states. Second, this within-state change was accompanied by a broadly convergent shift in New Jersey's wage mechanism toward other states'. Third, and most tellingly, this apparent structural change largely dissolves once benchmarked against a synthetic counterfactual constructed to match New Jersey's own pre-reform trajectory: New Jersey's post-reform mechanism proves statistically indistinguishable from what that trajectory would have predicted in the absence of the reform, whether or not the level of wages is included in the comparison.

The cosine metric proposed in this paper is not without limitations, both mathematical and economic. Its precision varies across its range: inference becomes unreliable for nearly identical mechanisms, and grows less sharp as mechanisms diverge. Yet these limitations matter less than they might appear, once the metric is considered against the applications it is designed for. Economic mechanisms are rarely exactly identical, and formally measuring the similarity of mechanisms already known to be fundamentally different offers little practical value. The economically interesting cases lie between these two extremes — in the judgment of whether two economies, markets, or groups respond to a shock similarly enough for one's experience to inform expectations about the other. It is precisely in this middle range that the metric is best suited, and precisely here that it performs best. Structural similarity analysis is not a substitute for a comprehensive economic assessment; it is a diagnostic tool that precedes one, helping to identify where deeper analysis is warranted and where lessons may plausibly travel across settings — a role relevant to questions ranging from investment decisions to policy transfer. Ultimately, the relevant question is not only whether two economies look alike, but whether they work alike, and whether they continue to do so when struck by the same shock — the question structural similarity analysis is proposed to answer.

\bibliographystyle{apalike}

\bibliography{bibliography}

@article{BonhommeManresa2015,
  author  = {Bonhomme, S. and Manresa, E.},
  title   = {Grouped Patterns of Heterogeneity in Panel Data},
  journal = {Econometrica},
  volume  = {83},
  number  = {3},
  pages   = {1147--1184},
  year    = {2015},
  note    = {doi: 10.3982/ECTA11319}
}

@article{Sims1980,
  author  = {Sims, C.},
  title   = {Macroeconomics and Reality},
  journal = {Econometrica},
  volume  = {48},
  number  = {1},
  pages   = {1--48},
  year    = {1980},
  note    = {doi: 10.2307/1912017}
}

@article{Keetal2015,
  author  = {Ke, T. and Fan, J. and Wu, Y.},
  title   = {Homogeneity Pursuit},
  journal = {Journal of the American Statistical Association},
  volume  = {110},
  number  = {509},
  pages   = {175--194},
  year    = {2015},
  note    = {doi: 10.1080/01621459.2014.892882}
}

@article{PattonWeller2023,
  author  = {Patton, A. J. and Weller, B. M.},
  title   = {Testing for Unobserved Heterogeneity via k-Means Clustering},
  journal = {Journal of Business \& Economic Statistics},
  volume  = {41},
  number  = {3},
  pages   = {737--751},
  year    = {2023},
  note    = {doi: 10.1080/07350015.2022.2061983}
}

@article{PesaranYamagata2008,
  author  = {Pesaran, M. H. and Yamagata, T.},
  title   = {Testing Slope Homogeneity in Large Panels},
  journal = {Journal of Econometrics},
  volume  = {142},
  number  = {1},
  pages   = {50--93},
  year    = {2008},
  note    = {doi: 10.1016/j.jeconom.2007.05.010}
}

@techreport{Heckmanetal2003,
  author      = {Heckman, James J. and Lochner, Lance J. and Todd, Petra E.},
  title       = {Fifty Years of {Mincer} Earnings Regressions},
  institution = {National Bureau of Economic Research},
  type        = {NBER Working Paper},
  number      = {9732},
  year        = {2003},
  month       = {May},
  note        = {doi: 10.3386/w9732}
}

@incollection{Lemieux2006,
  author    = {Lemieux, Thomas},
  title     = {The {Mincer Equation} Thirty Years After {Schooling, Experience, and Earnings}},
  booktitle = {Jacob Mincer: A Pioneer of Modern Labor Economics},
  editor    = {Grossbard, Shoshana},
  publisher = {Springer},
  address   = {New York},
  year      = {2006},
  pages     = {127--145},
  note      = {doi: 10.1007/0-387-29175-X\_11}
}

@book{Mincer1974,
  author    = {Mincer, Jacob A.},
  title     = {Schooling, Experience, and Earnings},
  year      = {1974},
  publisher = {National Bureau of Economic Research},
  address   = {New York},
  isbn      = {0-870-14265-8}
}

@article{CardKrueger1994,
  author  = {Card, David and Krueger, Alan B.},
  title   = {Minimum Wages and Employment: A Case Study of the Fast-Food Industry in {New Jersey} and {Pennsylvania}},
  journal = {American Economic Review},
  year    = {1994},
  volume  = {84},
  number  = {4},
  pages   = {772--793},
  month   = {September}
}

@article{Dubeetal2010,
  author  = {Dube, Arindrajit and Lester, T. William and Reich, Michael},
  title   = {Minimum Wage Effects Across State Borders: Estimates Using Contiguous Counties},
  journal = {The Review of Economics and Statistics},
  year    = {2010},
  volume  = {92},
  number  = {4},
  pages   = {945--964},
  note    = {doi: 10.1162/REST\_a\_00039}
}

@article{Cengizetal2019,
  author  = {Cengiz, Doruk and Dube, Arindrajit and Lindner, Attila and Zipperer, Ben},
  title   = {The Effect of Minimum Wages on Low-Wage Jobs},
  journal = {The Quarterly Journal of Economics},
  year    = {2019},
  volume  = {134},
  number  = {3},
  pages   = {1405--1454},
  note    = {doi: 10.1093/qje/qjz014}
}

@book{Becker1964,
  author    = {Becker, Gary S.},
  title     = {Human Capital: A Theoretical and Empirical Analysis, with Special Reference to Education},
  year      = {1964},
  publisher = {National Bureau of Economic Research},
  address   = {New York}
}

@article{Wied2026,
  author  = {Wied, Dominik},
  title   = {A Synthetic Control Approach to Conditional Distributional Treatment Effects},
  year    = {2026},
  journal = {arXiv preprint arXiv:2606.09625},
  note    = {doi: 10.48550/arXiv.2606.09625}
}

@techreport{Fisheretal2015,
  author      = {Fisher, Eric O'N. and Gilbert, John and Marshall, Kathryn G. and Oladi, Reza},
  title       = {A New Measure of Economic Distance},
  institution = {CESifo},
  type        = {CESifo Working Paper},
  number      = {5362},
  year        = {2015},
  address     = {Munich}
}

@article{Andrews1987,
  author  = {Andrews, Donald W. K.},
  title   = {Asymptotic Results for Generalized {Wald} Tests},
  journal = {Econometric Theory},
  year    = {1987},
  volume  = {3},
  number  = {3},
  pages   = {348--358},
  doi     = {10.1017/S0266466600010434}
}

@article{DufourEtAl2025,
  author  = {Dufour, Jean-Marie and Renault, Eric and Zinde-Walsh, Victoria},
  title   = {{Wald} Tests When Restrictions Are Locally Singular},
  journal = {The Annals of Statistics},
  year    = {2025},
  volume  = {53},
  number  = {2},
  pages   = {457--476},
  doi     = {10.1214/24-AOS2398}
}

\appendix

\section{Delta-Method Inference for Cosine Similarity}
\label{app:cosine_inference}

Consider two economic mechanisms, represented within a common linear modelling framework by the coefficient vectors $\boldsymbol{\beta_1}$
 and $\boldsymbol{\beta_2}$, respectively, whose structural similarity is measured using the cosine metric. Let
$\widehat{\boldsymbol{\beta}}_1$ and
$\widehat{\boldsymbol{\beta}}_2$ denote two $k\times 1$ estimators of
$\boldsymbol{\beta}_1$ and $\boldsymbol{\beta}_2$, respectively, and suppose
that
\begin{equation*}
\sqrt{n}
\begin{pmatrix}
\widehat{\boldsymbol{\beta}}_1-\boldsymbol{\beta}_1 \\
\widehat{\boldsymbol{\beta}}_2-\boldsymbol{\beta}_2
\end{pmatrix}
\xrightarrow{d}
\mathcal{N}
\left(
\boldsymbol{0},
\boldsymbol{\Sigma}
\right),
\qquad
\boldsymbol{\Sigma}
=
\begin{pmatrix}
\boldsymbol{\Sigma}_{11} & \boldsymbol{\Sigma}_{12} \\
\boldsymbol{\Sigma}_{21} & \boldsymbol{\Sigma}_{22}
\end{pmatrix}.
\label{eq:app_joint_an}
\end{equation*} \\
The cosine similarity between the two coefficient vectors is
\begin{equation*}
g(\boldsymbol{\beta}_1,\boldsymbol{\beta}_2)
=
\frac{
\boldsymbol{\beta}_1^{\prime}\boldsymbol{\beta}_2
}{
\|\boldsymbol{\beta}_1\|\,\|\boldsymbol{\beta}_2\|
},
\label{eq:app_cosine}
\end{equation*}
with estimator
\begin{equation*}
\widehat g
=
\frac{
\widehat{\boldsymbol{\beta}}_1^{\prime}
\widehat{\boldsymbol{\beta}}_2
}{
\|\widehat{\boldsymbol{\beta}}_1\|\,
\|\widehat{\boldsymbol{\beta}}_2\|
}.
\label{eq:app_cosine_est}
\end{equation*}\\
Let
\begin{equation*}
a=\|\boldsymbol{\beta}_1\|,
\qquad
b=\|\boldsymbol{\beta}_2\|,
\qquad
\boldsymbol{u}_1=\frac{\boldsymbol{\beta}_1}{a},
\qquad
\boldsymbol{u}_2=\frac{\boldsymbol{\beta}_2}{b},
\end{equation*}
so that
\[
g=\boldsymbol{u}_1^{\prime}\boldsymbol{u}_2.
\]\\
Provided that
$\boldsymbol{\beta}_1\neq\boldsymbol{0}$ and
$\boldsymbol{\beta}_2\neq\boldsymbol{0}$,
the cosine similarity is continuously differentiable. Its gradients are
\begin{align*}
\frac{\partial g}{\partial\boldsymbol{\beta}_1}
&=
\frac{\boldsymbol{\beta}_2}{ab}
-
g\frac{\boldsymbol{\beta}_1}{a^2}
=
\frac{1}{a}
\left(
\boldsymbol{u}_2-g\boldsymbol{u}_1
\right),
\label{eq:app_grad1}
\\
\frac{\partial g}{\partial\boldsymbol{\beta}_2}
&=
\frac{\boldsymbol{\beta}_1}{ab}
-
g\frac{\boldsymbol{\beta}_2}{b^2}
=
\frac{1}{b}
\left(
\boldsymbol{u}_1-g\boldsymbol{u}_2
\right).
\end{align*}
Define the stacked gradient
\begin{equation*}
\nabla g
=
\begin{pmatrix}
\dfrac{\partial g}{\partial\boldsymbol{\beta}_1}
\\[0.8em]
\dfrac{\partial g}{\partial\boldsymbol{\beta}_2}
\end{pmatrix}.
\label{eq:app_gradient}
\end{equation*}\\
The multivariate delta method then gives
\begin{equation*}
\sqrt{n}
\left[
g(\widehat{\boldsymbol{\beta}}_1,
  \widehat{\boldsymbol{\beta}}_2)
-
g(\boldsymbol{\beta}_1,\boldsymbol{\beta}_2)
\right]
\xrightarrow{d}
\mathcal{N}(0,V_g),
\label{eq:app_delta}
\end{equation*}
where
\begin{equation*}
V_g
=
\nabla g^{\prime}
\boldsymbol{\Sigma}
\nabla g.
\label{eq:app_var_compact}
\end{equation*}\\
Substituting the gradients yields
\begin{align*}
V_g
={}&
\frac{1}{a^2}
(\boldsymbol{u}_2-g\boldsymbol{u}_1)^{\prime}
\boldsymbol{\Sigma}_{11}
(\boldsymbol{u}_2-g\boldsymbol{u}_1)
\nonumber\\
&+
\frac{2}{ab}
(\boldsymbol{u}_2-g\boldsymbol{u}_1)^{\prime}
\boldsymbol{\Sigma}_{12}
(\boldsymbol{u}_1-g\boldsymbol{u}_2)
\nonumber\\
&+
\frac{1}{b^2}
(\boldsymbol{u}_1-g\boldsymbol{u}_2)^{\prime}
\boldsymbol{\Sigma}_{22}
(\boldsymbol{u}_1-g\boldsymbol{u}_2).
\label{eq:app_var_expanded}
\end{align*}\\
A consistent estimator of the variance is obtained by replacing the
population quantities with their sample counterparts:
\begin{equation*}
\widehat{\operatorname{Var}}(\widehat g)
=
\widehat{\nabla g}^{\,\prime}
\widehat{\operatorname{Var}}
\begin{pmatrix}
\widehat{\boldsymbol{\beta}}_1 \\
\widehat{\boldsymbol{\beta}}_2
\end{pmatrix}
\widehat{\nabla g}.
\label{eq:app_var_est}
\end{equation*}
Hence,
\begin{equation*}
\widehat{\operatorname{se}}(\widehat g)
=
\sqrt{
\widehat{\operatorname{Var}}(\widehat g)
},
\end{equation*}
and an asymptotic $(1-\alpha)$ confidence interval is
\begin{equation*}
\left[
\widehat g
-
z_{1-\alpha/2}
\widehat{\operatorname{se}}(\widehat g),
\;
\widehat g
+
z_{1-\alpha/2}
\widehat{\operatorname{se}}(\widehat g)
\right].
\label{eq:app_ci}
\end{equation*}

\section{Monte Carlo tables}

\begin{table}[H]
\centering
\caption{Finite-Sample Performance of the Cosine Similarity Estimator}
\label{tab:mc_results}

\begin{tabular}{cccccccc}
\hline\hline
$N$
& $\theta_0$
& Mean $\hat{\theta}$
& Bias
& RMSE
& MC SD
& Mean SE
& Coverage \\
\hline

50   & 0.00 &  0.0028 &  0.0028 & 0.2004 & 0.2003 & 0.2040 & 0.9371 \\
200  & 0.00 & -0.0016 & -0.0016 & 0.1008 & 0.1008 & 0.1005 & 0.9444 \\
500  & 0.00 &  0.0007 &  0.0007 & 0.0631 & 0.0631 & 0.0634 & 0.9494 \\
1000 & 0.00 & -0.0005 & -0.0005 & 0.0456 & 0.0456 & 0.0447 & 0.9478 \\

\addlinespace

50   & 0.50 & 0.4792 & -0.0208 & 0.1795 & 0.1783 & 0.1781 & 0.9321 \\
200  & 0.50 & 0.4951 & -0.0049 & 0.0865 & 0.0863 & 0.0871 & 0.9482 \\
500  & 0.50 & 0.4994 & -0.0006 & 0.0550 & 0.0550 & 0.0549 & 0.9467 \\
1000 & 0.50 & 0.4989 & -0.0011 & 0.0390 & 0.0390 & 0.0388 & 0.9491 \\

\addlinespace

50   & 0.90 & 0.8608 & -0.0392 & 0.1065 & 0.0990 & 0.0980 & 0.9289 \\
200  & 0.90 & 0.8908 & -0.0092 & 0.0456 & 0.0446 & 0.0448 & 0.9420 \\
500  & 0.90 & 0.8963 & -0.0037 & 0.0281 & 0.0278 & 0.0279 & 0.9474 \\
1000 & 0.90 & 0.8981 & -0.0019 & 0.0194 & 0.0193 & 0.0196 & 0.9504 \\

\addlinespace

50   & 0.99 & 0.9473 & -0.0427 & 0.0674 & 0.0522 & 0.0596 & 0.9910 \\
200  & 0.99 & 0.9801 & -0.0099 & 0.0198 & 0.0172 & 0.0182 & 0.9566 \\
500  & 0.99 & 0.9860 & -0.0040 & 0.0107 & 0.0099 & 0.0099 & 0.9333 \\
1000 & 0.99 & 0.9881 & -0.0019 & 0.0069 & 0.0066 & 0.0066 & 0.9385 \\

\hline\hline
\end{tabular}

\vspace{0.5em}

\begin{minipage}{0.94\textwidth}
\footnotesize
\textit{Notes:} The table reports Monte Carlo results based on
$10{,}000$ replications. $\theta_0$ denotes the true cosine similarity
between the two population coefficient vectors. Mean $\hat{\theta}$
is the average estimated cosine similarity across Monte Carlo
replications. Bias is defined as
$\mathbb{E}[\hat{\theta}]-\theta_0$, and RMSE as
$\sqrt{\mathbb{E}[(\hat{\theta}-\theta_0)^2]}$.
MC SD denotes the empirical standard deviation of $\hat{\theta}$
across Monte Carlo replications, while Mean SE is the average
standard error obtained from the proposed delta-method variance
estimator. Coverage reports the empirical coverage probability of
the nominal 95\% confidence interval.
\end{minipage}

\end{table}

\end{document}